\documentclass[sigconf,nonacm]{acmart}

\newif\ifcleanfirstpage
\cleanfirstpagetrue
\ifcleanfirstpage
  \setcopyright{none}
  \makeatletter
  \renewcommand{\footnotetextcopyrightpermission}[1]{}
  \makeatother
\fi

\usepackage{amsmath,mathtools}
\usepackage{algorithm}
\usepackage[noend]{algpseudocode}
\usepackage{booktabs}
\usepackage{array}
\usepackage{tabularx}
\usepackage{microtype}
\usepackage{enumitem}
\usepackage{url}
\usepackage{graphicx}
\newtheorem{theorem}{Theorem}
\newtheorem{lemma}[theorem]{Lemma}
\newtheorem{proposition}[theorem]{Proposition}

\newcommand{\OPT}{\mathrm{OPT}}
\newcommand{\calM}{\mathcal{M}}
\newcommand{\E}{\mathbb{E}}

\newcommand{\one}{\mathbf{1}}
\newcommand{\W}{\mathrm{W}}

\title{A Truthful-in-Expectation 2-Approximation for Pair-Specific Cascade Auctions in AI-Generated Content}
\author{Tao Jiang\textsuperscript{\dag}\quad Minbo Gao\textsuperscript{\dag}\quad Shaowei Cai}
\affiliation{%
  \department{Key Laboratory of System Software (Chinese Academy of Sciences)}
  \department{State Key Laboratory of Computer Science}
  \institution{Institute of Software, Chinese Academy of Sciences}
  \department{School of Computer Science and Technology}
  \institution{University of Chinese Academy of Sciences}
  \city{Beijing}
  \country{China}}
\email{\{jiangt,gaomb,caisw\}@ios.ac.cn}

\renewcommand{\shortauthors}{Jiang, Gao, and Cai}

\begin{document}

\begin{abstract}
Generative AI interfaces can insert sponsored content at multiple points in an answer, where click-through rates (CTRs) depend on the advertiser--position pair. Under the cascade model, users read ads sequentially and stop after their first click, so both ad assignment and display order affect welfare. For arbitrary pair-specific CTRs, Balseiro et al.\ (WWW~2026) gave a $4/(1-\varepsilon)$-approximation whose monotonicity has not been established, and a monotone $O(\log m)$-approximation for $m$ candidate positions. Earlier truthful constant-factor mechanisms require factorized CTR models and do not apply. Whether monotonicity is compatible with a constant-factor approximation in this general setting remained open.

We answer this question affirmatively. For arbitrary pair-specific CTRs and any limit on displayed ads, we give a polynomial-time mechanism that is truthful and individually rational in expectation and achieves a $2$-approximation to optimal cascade welfare. Our mechanism selects an advertiser--position matching and a designated last ad, then randomly thins earlier ads so that each selected ad's expected cascade CTR is the same fraction of its standalone CTR. This yields an exact closed-form expression for expected welfare, which we maximize over a bid-independent lottery range using linear programming. Exact optimization ensures monotonicity, while a random-prefix decomposition proves the factor-two guarantee, which is tight for this range. We also show that exactly optimizing the previous welfare surrogate can violate monotonicity, and evaluate welfare and running time on synthetic instances.
\end{abstract}

\ccsdesc[500]{Theory of computation~Algorithmic game theory and mechanism design}
\ccsdesc[300]{Information systems~Sponsored search advertising}
\keywords{online advertising, position auctions, generative AI, cascade model, truthful-in-expectation mechanisms, approximation algorithms}

\maketitle
\hypersetup{pdfauthor={Tao Jiang, Minbo Gao, Shaowei Cai}}
{\renewcommand{\thefootnote}{\dag}\footnotetext{Tao Jiang and Minbo Gao contributed equally to this work.}}

\section{Introduction}

Generative search answers, AI-written travel guides, and shopping summaries can mention sponsors at several points in the generated text. In such content, the click-through rate (CTR) of an ad may depend on both the advertiser and the position where the ad appears. Balseiro et al.~\cite{balseiro2026position} model this setting with a separate CTR $p_{ij}$ for each advertiser $i$ and position $j$, so these rates need not split into an advertiser effect and a position effect as they do in classical position auctions~\cite{varian2007position,edelman2007gsp}. Users read generated text in order, and in the cascade model~\cite{craswell2008experimental,kempe2008cascade,aggarwal2008markovian} they stop reading after their first click. As a result, each displayed ad takes clicks away from the ads shown after it. Like Balseiro et al., we assume that the platform decides the order in which users see the matched ads and that this order does not change their standalone CTRs; Section~\ref{sec:discussion} discusses layouts with a fixed order. Advertisers therefore compete both for positions and for the user's attention as the user reads.

An auction in this setting must do more than maximize welfare: it must also give advertisers a reason to bid their true value per click. No known polynomial-time algorithm finds the welfare-maximizing allocation, so the platform cannot run the Vickrey--Clarke--Groves (VCG) mechanism with exact optimization. When each advertiser holds a single private value, truthful bidding requires that an advertiser receive at least as many expected clicks when it raises its bid~\cite{myerson1981optimal,archer2001truthful}. Approximation algorithms with strong welfare guarantees need not satisfy this condition, so the main challenge is to achieve efficient computation, good welfare, and monotone clicks at the same time.

Existing algorithms for this model achieve one of these goals but not both. Balseiro et al.~\cite{balseiro2026position} obtain a $4/(1-\varepsilon)$-approximation, for any fixed $\varepsilon>0$, by optimizing a simpler restricted-welfare objective with a polynomial-time approximation scheme for budgeted matching~\cite{berger2008budgeted}, but it remains unknown whether the resulting allocation is monotone. Their monotone algorithm runs a greedy matching inside a randomly chosen CTR bucket, and it loses an $O(\log m)$ factor for $m$ positions. Earlier work gives truthful constant-factor auctions for structured cascade models~\cite{gatti2013mechanism,farina2017cascade}, but these results rely on CTRs that factor into simpler terms and do not cover arbitrary $p_{ij}$. Balseiro et al.~\cite{balseiro2026position} therefore asked the following question: \emph{can a polynomial-time monotone allocation achieve a constant-factor approximation to welfare for arbitrary pair-specific cascade CTRs?}

We answer this question affirmatively; Section~\ref{sec:truth} states the formal version of the following result (Theorem~\ref{thm:main-formal}).
\begin{theorem}[Main result, informal]\label{thm:main}
For arbitrary public pair-specific CTRs and any limit on the number of displayed ads, there is a polynomial-time mechanism that is truthful and individually rational in expectation, and whose expected welfare is at least half of the optimal cascade welfare.
\end{theorem}

Our mechanism optimizes the expected welfare that the displayed ads actually deliver, rather than a surrogate. It does so over a fixed range of randomized displays, that is, a set of lotteries that does not depend on the bids. We need such a range because exact optimization over all cascade displays is hard, and Section~\ref{sec:counterexample} shows that even an exact optimum of a surrogate can violate monotonicity. Each lottery in our range picks a matching and one ad to show last, and it keeps each earlier ad independently with a carefully chosen probability. These probabilities give every selected advertiser an expected cascade CTR equal to its standalone CTR times a common factor, so we can write the expected welfare of each lottery in closed form.

For a fixed last ad, this welfare is a ratio of two linear functions of the matching, and one linear program maximizes it exactly. The linear program always has an optimal solution that is a matching, so relaxing the matching constraints loses nothing (Theorem~\ref{thm:poly}). Since the range does not depend on the bids and our objective equals the delivered welfare, exact optimization makes the allocation monotone, and standard payments then make the mechanism truthful in expectation~\cite{lavi2011truthful,dobzinski2013power} (Theorem~\ref{thm:mono}, Proposition~\ref{prop:tie}).

We also show that restricting the mechanism to this range costs at most a factor of $2$ in welfare. The proof writes the welfare of any cascade display as the expected standalone welfare of a random prefix of its ads and compares this quantity with a matching-based benchmark (Theorem~\ref{thm:approx}). The factor $2$ is tight for our range (Proposition~\ref{prop:tight}), but it is not a lower bound for truthful mechanisms in general. In addition, we show that optimizing the earlier surrogate exactly does not restore monotonicity (Section~\ref{sec:counterexample}), and we evaluate running time, welfare, and sensitivity to CTR errors on synthetic instances (Section~\ref{sec:synthetic}). Theorem~\ref{thm:main} thus improves the best previous factor for monotone allocations from $O(\log m)$ to $2$, with no restrictions on the number of displayed ads or on the CTR matrix.

\subsection{Related Work}\label{sec:related}

\begin{table*}[t]
\centering
\caption{Incentive and approximation guarantees for related cascade auctions under different click models.}
\label{tab:landscape}
\small
\setlength{\tabcolsep}{5pt}
\renewcommand{\arraystretch}{1.15}
\begin{tabularx}{\textwidth}{@{}>{\raggedright\arraybackslash}X>{\raggedright\arraybackslash}p{0.28\textwidth}>{\centering\arraybackslash}p{0.14\textwidth}@{}}
\toprule
Work & Incentive guarantee & Factor $C$ \\
\midrule
\multicolumn{3}{@{}l}{\textbf{Structured cascade models (APDC; factorized CTRs)}} \\
\addlinespace[2pt]
\hspace{0.7em}Gatti--Rocco~\cite{gatti2013mechanism}
& monotone; TIE
& $4/(1-\varepsilon)$ \\
\hspace{0.7em}Farina--Gatti~\cite{farina2017cascade} (equal slot prominences)
& MIR; DSIC
& $2$ \\
\addlinespace[6pt]
\multicolumn{3}{@{}l}{\textbf{Pair-specific first-click cascade (arbitrary $p_{ij}$)}} \\
\addlinespace[2pt]
\hspace{0.7em}Balseiro et al.~\cite{balseiro2026position} (restricted welfare)
& not established
& $4/(1-\varepsilon)$ \\
\hspace{0.7em}Balseiro et al.~\cite{balseiro2026position} (CTR buckets)
& monotone
& $O(\log m)$ \\
\hspace{0.7em}\textbf{This paper}
& \textbf{TIE; IR in expectation}
& \textbf{2} \\
\bottomrule
\end{tabularx}
\par\smallskip
{\footnotesize\raggedright
$C$: approximation factor (welfare at least $\OPT/C$). APDC: ad/position-dependent cascade; TIE: truthful in expectation; IR: individually rational; MIR: maximal in range; DSIC: dominant-strategy incentive compatible. Our factor $2$ is tight for the designated-last range.\par}
\end{table*}

\paragraph{Position auctions and cascade models.} Classical position auctions assume CTRs that factor into advertiser and position effects~\cite{varian2007position,edelman2007gsp}, and cascade and Markovian click models capture competition for sequential attention~\cite{craswell2008experimental,kempe2008cascade,aggarwal2008markovian}. For the ad/position-dependent cascade (APDC), in which a CTR is the product of an ad quality, a slot factor, and continuation probabilities that depend only on the ads above, Gatti and Rocco~\cite{gatti2013mechanism} and Farina and Gatti~\cite{farina2017cascade} give truthful constant-factor mechanisms (Table~\ref{tab:landscape}). In our model, $p_{ij}$ is an arbitrary matrix that also determines the continuation probability $1-p_{ij}$; APDC cannot represent even $p=\mathrm{diag}(\tfrac12,\tfrac12)$. Balseiro et al.~\cite{balseiro2026position} introduced pair-specific CTRs for AI-generated content; we adopt their cascade model. Other page-level externalities include spatial proximity~\cite{aggarwal2026spatial} and ad sizes under a page capacity~\cite{batziou2026capacity}, neither with first-click stopping.
\paragraph{Truthful randomized mechanisms.} Exact welfare maximization over a bid-independent distributional range is truthful in expectation~\cite{lavi2011truthful,dobzinski2013power}, and convex rounding~\cite{dughmi2016convex} applies this principle to roundings whose expected welfare is tractable. Attenuation~\cite{brubach2020attenuate} and contention resolution~\cite{chekuri2014cr,nuti2025matchingcr} also discard candidates to control interference, but they guarantee lower bounds on survival probabilities; our range yields exact expected cascade CTRs, which permits exact welfare optimization. Mechanisms for generative advertising use other allocation primitives~\cite{duetting2024mechanism,dubey2024auctions,xu2026adinsertion,jiang2026twomarkets,han2026quality,fang2026osda,elahi2026mirt}, and click-intent estimation~\cite{turnersmith2026pricing} is complementary to our assumption that $p_{ij}$ is given.

\paragraph{Concurrent and independent work.} We became aware of concurrent and independently developed work by Liu et al.~\cite{liu2026truthful}, first posted on arXiv on October 22, 2026. Both works establish polynomial-time truthful-in-expectation $2$-approximation mechanisms for pair-specific cascade auctions using randomized thinning and exact welfare optimization. The contributions differ in two respects. First, Liu et al.\ develop order-invariant thinning that accommodates a fixed editorial order, whereas our mechanism exploits control over display order through a designated-last lottery with an exactly optimizable expected-welfare objective. Second, we provide an explicit counterexample showing that even exact optimization of the restricted-welfare surrogate of Balseiro et al.~\cite{balseiro2026position} can violate monotonicity of cascade click probabilities (Section~\ref{sec:counterexample}). This differs from the counterexample in Liu et al., which concerns direct, unthinned rendering of an odds-optimized allocation. Our results were developed independently.

\section{Model}\label{sec:model}

There are $n$ advertisers $N=[n]$ and $m$ candidate positions $P=[m]$. Advertiser $i$ has a private value $v_i\ge0$ per click. A \emph{matching} $M\subseteq N\times P$ assigns each advertiser at most one position and each position at most one advertiser; at most $K\le m$ ads can be shown, and $\calM_K$ denotes the set of such matchings, including $\emptyset$. A \emph{display} is an ordered list of the edges of a matching in $\calM_K$. We write $e=(i,j)$ for an edge, $p_e=p_{ij}$, and $v_e=v_i$. Pairs that the platform excludes, for instance for brand safety, are removed from the bipartite graph, and all results apply to the remaining graph.

\paragraph{Display order.} As in prior work on this model~\cite{balseiro2026position}, the platform chooses the order in which the user examines the matched ads, and standalone CTRs do not depend on this order. We assume that the interface permits this; Algorithm~1 uses it to show one ad last. Layouts whose placements cannot be reordered require additional constraints (Section~\ref{sec:discussion}).

\paragraph{Click probabilities.} We distinguish three quantities.
\begin{itemize}[leftmargin=1.3em,itemsep=1pt,topsep=2pt]
\item The \textbf{standalone CTR} $p_{ij}\in[0,1]$ is the probability that ad $i$ is clicked at position $j$ when it is shown alone. It is fixed by the platform before bids are evaluated, for example as the output of a prediction model. The theorem assumes that this public matrix is the same CTR matrix governing user click behavior; Section~\ref{sec:noise} studies welfare under misspecification.
\item The \textbf{cascade CTR} $\pi_i(A)$ is advertiser $i$'s click probability given a realized display $A$; it is a probability, not a click indicator. If $A$ shows $e_1,\dots,e_\ell$ in this order, then
\begin{equation}\label{eq:cascade-edge}
  \pi_{e_k}(A)=p_{e_k}\prod_{h<k}(1-p_{e_h}),
\end{equation}
since the user reaches $e_k$ only if no earlier ad is clicked. Here and below, $\pi_e(A)=\pi_i(A)$ for a displayed edge $e=(i,j)$.
\item The \textbf{expected cascade CTR} $q_i=\E_{A\sim D}[\pi_i(A)]$ is the expected allocation of advertiser $i$ when a mechanism draws $A$ from a distribution $D$.
\end{itemize}
The \emph{cascade welfare} of a display $A$ is $\W(A)=\sum_iv_i\pi_i(A)$, and $\OPT(v)$ is the maximum cascade welfare over all $M\in\calM_K$ and all display orders.

\paragraph{Incentives.} We assume that advertisers are risk-neutral with quasi-linear expected utility $v_iq_i-T_i$, where $T_i$ is the expected payment, and that only $v_i$ is private; advertisers do not report $p$. A randomized allocation rule is \emph{monotone} if $q_i$ is nondecreasing in $v_i$ for every $i$ and $v_{-i}$. This is the standard condition for truthful single-parameter mechanisms~\cite{myerson1981optimal,archer2001truthful}, and it is the notion of monotonicity used in~\cite{balseiro2026position}.

\section{The Designated-Last Mechanism}\label{sec:mechanism}

We now define the lottery range. For a matching $M$, let
\[
  S(M)=\sum_{e\in M}p_e,\qquad A_v(M)=\sum_{e\in M}v_ep_e
\]
denote its total standalone CTR and its standalone welfare. For nonempty $M$ and a designated edge $d\in M$, define
\begin{equation}\label{eq:scale}
 c(M,d)=\frac{1}{1+S(M)-p_d},\qquad
 G_v(M,d)=\frac{A_v(M)}{1+S(M)-p_d},
\end{equation}
so that $G_v(M,d)=c(M,d)\,A_v(M)$.

For a fixed matching $M$, the numerator $A_v(M)$ does not depend on $d$, and the denominator decreases in $p_d$. Hence
\[
 d^\star(M)\in\arg\max_{e\in M}p_e:
\]
the best designated edge is a selected ad with the largest standalone CTR, independently of the bids, and the mechanism effectively chooses a matching.

\begin{algorithm}[t]
\caption{Designated-Last Cascade Mechanism}
\label{alg:dl}
\begin{algorithmic}[1]
\State Compute an exact maximizer $(M^\star,d^\star)$ of $G_v(M,d)$ over nonempty $M\in\calM_K$ and $d\in M$ (Section~\ref{sec:opt}).
\If{$G_v(M^\star,d^\star)=0$}
  \State \Return an empty display
\EndIf
\State Order $M^\star\setminus\{d^\star\}$ by a fixed public position order; append $d^\star$.
\State Denote the sequence by $e_1,\dots,e_\ell$ and set $c=c(M^\star,d^\star)$.
\For{$k=1,\dots,\ell$}
  \State Set $R_k=1-c\sum_{h<k}p_{e_h}$.
  \State Retain $e_k$ independently with probability $a_k=c/R_k$.
\EndFor
\State Display the retained ads in order and charge the payments of Section~\ref{sec:truth}.
\end{algorithmic}
\end{algorithm}

The retention probabilities depend only on the public CTRs and the selected pair, and every display is a subset of a matching in $\calM_K$. Figure~\ref{fig:lottery} illustrates the thinning and click process.

\begin{figure}[t]
\centering
\includegraphics[width=\columnwidth]{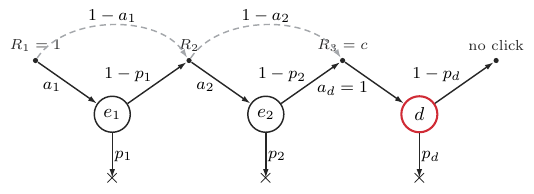}
\Description{A branching diagram for the selected sequence e1, e2, and designated-last d. Dashed arcs skip e1 and e2 with probabilities one minus a1 and one minus a2. Solid branches retain each ad with probability a, then either continue with probability one minus p or end at a cross with click probability p. Reach probabilities are R1 equals one, R2, and R3 equals c. The designated ad d is always retained and has a red outline. Passing d without a click ends at the no-click state.}
\caption{Designated-last thinning for a three-ad matching. Non-designated ads are independently retained with probabilities $a_k=c/R_k$; dashed arcs skip dropped ads. If reached, a retained ad is clicked with probability $p_k$, ending the cascade. Ad $d$ is always retained ($a_d=1$).}
\label{fig:lottery}
\end{figure}

\subsection{An exact cascade lottery}
Fix a nonempty matching $M$, an edge $d\in M$, and an order $e_1,\dots,e_\ell$ of $M$ with $e_\ell=d$; write $p_k=p_{e_k}$ and $c=c(M,d)$. Every set of edges that excludes $d$ has CTR mass at most $S(M)-p_d$, so $R_k\ge 1-c(S(M)-p_d)=c$. Hence $a_k=c/R_k\in[0,1]$ and $a_\ell=1$. Since $1/c=1+S(M)-p_d$, the retention probabilities have the suffix form
\begin{equation}\label{eq:suffix}
 a_k=\frac{1}{1+\sum_{h=k}^{\ell-1}p_h}:
\end{equation}
ad $e_k$ is kept with probability one over one plus the CTR mass of $e_k$ and the non-designated ads after it (the sum is empty for $k=\ell$). Let $D_{M,d}$ denote the resulting distribution over displays.

\begin{lemma}[Exact cascade lottery]\label{lem:lottery}
Under $D_{M,d}$, every $e_k\in M$ has expected cascade CTR $\E[\pi_{e_k}]=c(M,d)\,p_k$. Consequently, the expected cascade welfare of $D_{M,d}$ is $G_v(M,d)$, and the expected allocation does not depend on the order of $M\setminus\{d\}$.
\end{lemma}

\begin{proof}
Let $\widehat R_k$ be the probability that the user reaches position $k$ of the sequence, that is, that no retained earlier ad is clicked. Retention and clicks are independent across ads, so $\widehat R_{k+1}=\widehat R_k(1-a_kp_k)$. If $\widehat R_k=R_k$, then $\widehat R_{k+1}=R_k-cp_k=R_{k+1}$; since $\widehat R_1=R_1=1$, induction gives $\widehat R_k=R_k$ for all $k$. Therefore $\E[\pi_{e_k}]=a_k\,p_k\,\widehat R_k=(c/R_k)\,p_k\,R_k=c\,p_k$, and summing $v_e\,\E[\pi_e]$ over $e\in M$ gives $G_v(M,d)$.
\end{proof}

For the three-ad matching of Figure~\ref{fig:lottery} with $p=(0.5,0.4,0.6)$, deterministic display gives cascade CTRs $(0.5,0.2,0.18)$, whereas the lottery gives $\tfrac{10}{19}(0.5,0.4,0.6)$. Thinning sacrifices some clicks, but the expected welfare of every lottery in the range has a closed form for all bids, and this is the quantity that the mechanism maximizes. Attenuation schemes guarantee lower bounds on survival probabilities~\cite{brubach2020attenuate}; here the identity $\E[\pi_e]=c(M,d)\,p_e$ is exact.

Among lotteries that fix the last ad in advance, the scale $c(M,d)$ cannot be improved.

\begin{proposition}[Largest common scale under a fixed designated-last constraint]\label{prop:maxscale}
Fix $M$ and $d\in M$ with $p_d>0$. Consider any lottery over displays that use only edges of $M$ in which every realization containing $d$ places $d$ after all other displayed edges. If every $e\in M$ has $\E[\pi_e]=c\,p_e$ for a common $c$, then $c\le c(M,d)$.
\end{proposition}

\begin{proof}
Since the user leaves after the first click, clicks on different ads are disjoint events. Ad $d$ can be clicked only if no other displayed ad is clicked, and then with probability at most $p_d$. Hence $c\,p_d=\E[\pi_d]\le p_d\bigl(1-\sum_{e\ne d}c\,p_e\bigr)=p_d\bigl(1-c(S(M)-p_d)\bigr)$, which rearranges to $c\le c(M,d)$.
\end{proof}

Two remarks delimit the range. First, the symmetric scale $1/(1+S(M))$, which needs no designated edge, is never larger:
\begin{equation}\label{eq:symmetric}
 F_v(M)=\frac{A_v(M)}{1+S(M)}\le G_v(M,d)\qquad\text{for every }d\in M.
\end{equation}
We use $F_v$ only as a benchmark in Section~\ref{sec:approx}. Second, Proposition~\ref{prop:maxscale} does not apply to lotteries that randomize which ad is last. Showing two ads with $p=\tfrac12$ in uniformly random order gives each an expected cascade CTR of $\tfrac34p$, which exceeds the designated-last scale $\tfrac23$. Whether such ranges admit exact polynomial-time optimization is open (Section~\ref{sec:discussion}).

\section{Exact Optimization in Polynomial Time}\label{sec:opt}

Step~1 of Algorithm~1 maximizes $G_v$ over all pairs $(M,d)$. We fix each of the at most $nm$ designated edges $d$ in turn and maximize a ratio of linear functions over the matchings that contain $d$. The Charnes--Cooper transformation linearizes the ratio, and a convex-combination argument shows that the resulting LP relaxation is exact.

Let $Q_K=\{x\ge0:\sum_jx_{ij}\le1\ \forall i,\ \sum_ix_{ij}\le1\ \forall j,\ \sum_{i,j}x_{ij}\le K\}$ be the cardinality-constrained bipartite matching polytope. It is integral. To see this, consider the network with a source $o$, a node for each advertiser and each position, and a sink $\omega$, with arcs $(o,i)$, $(i,j)$, and $(j,\omega)$ of capacity $1$ and a return arc $(\omega,o)$ of capacity $K$. Every $x\in Q_K$ extends to a feasible circulation by sending $\sum_jx_{ij}$ on $(o,i)$, $\sum_ix_{ij}$ on $(j,\omega)$, and $\sum_{i,j}x_{ij}$ on $(\omega,o)$, and every feasible circulation restricts to a point of $Q_K$. Circulation polytopes with integral capacities are integral, and integral circulations restrict to matchings with at most $K$ edges. If the positions are partitioned into sections $g$ with caps $K_g$, one node per section with capacity $K_g$ is added between its positions and $\omega$; integrality is preserved, and subsets of feasible sets remain feasible, so Lemma~\ref{lem:prefix} and Theorem~\ref{thm:approx} also apply unchanged. Its face $Q_K(d)=\{x\in Q_K:x_d=1\}$ is therefore integral, and its vertices are the matchings in $\calM_K$ that contain $d$. For $x\in Q_K(d)$ let
\begin{equation}\label{eq:fixed-d-ratio}
 G_{v,d}(x)=\frac{\sum_e v_e p_e x_e}{1+\sum_{e\ne d}p_e x_e}.
\end{equation}
Since the denominator is at least $1$, the substitution $z=tx$, $t=1/(1+\sum_{e\ne d}p_ex_e)$~\cite{charnes1962programming} yields the polynomial-size linear program
\begin{equation}\label{eq:lp}
\begin{aligned}
\max\ &\textstyle\sum_e v_ep_ez_e\\
\text{s.t. }&\textstyle\sum_jz_{ij}\le t\ \forall i,\quad \sum_iz_{ij}\le t\ \forall j,\quad \sum_ez_e\le Kt,\\
&\textstyle z_d=t,\quad t+\sum_{e\ne d}p_ez_e=1,\quad z,t\ge0.
\end{aligned}
\end{equation}
Every feasible $(z,t)$ has $t>0$, because $t=0$ forces $z=0$ and violates the normalization. Hence $x=z/t\in Q_K(d)$ has the same objective value, and conversely every $x\in Q_K(d)$ gives a feasible point.

\paragraph{No fractional gap.} Let $x=\sum_{M\ni d}\xi_M\one_M$ be a convex combination of incidence vectors $\one_M$ of matchings that contain $d$, and let $\Delta_M=1+S(M)-p_d>0$. Because $A_v$ and $S$ are linear,
\begin{equation}\label{eq:convex-ratio}
\begin{aligned}
G_{v,d}(x)&=\frac{\sum_M\xi_M A_v(M)}{\sum_M\xi_M \Delta_M}
=\sum_{M\ni d}\underbrace{\frac{\xi_M\Delta_M}{\sum_{M'}\xi_{M'}\Delta_{M'}}}_{\mu_M}\,G_v(M,d)\\
&\le\max_{M\ni d}G_v(M,d).
\end{aligned}
\end{equation}
The fractional objective is thus a convex combination of integral objectives, and the LP optimum $r_d$ equals the best value of a matching that contains $d$.

\paragraph{Recovery.} For $M\ni d$, the inequality $G_v(M,d)\le r_d$ is equivalent to $\sum_{e\in M\setminus\{d\}}(v_e-r_d)p_e\le r_d-v_dp_d$, with equality exactly for optimal $M$. Hence a maximum-weight matching with at most $K-1$ edges, in the graph without $d$'s advertiser and position and with weights $(v_e-r_d)p_e$, together with $d$ is optimal.

\begin{theorem}[Exact optimization]\label{thm:poly}
For rational inputs, a maximizing pair $(M^\star,d^\star)$ and the value
\begin{equation}\label{eq:phidef}
 \Phi(v)=\max\Bigl\{0,\ \max_{M\in\calM_K\setminus\{\emptyset\},\,d\in M}G_v(M,d)\Bigr\}
\end{equation}
are computable in deterministic polynomial time, and so are the expected allocation $q(v)$ and the expected payments of Section~\ref{sec:truth}.
\end{theorem}

\begin{proof}
For each of the at most $nm$ edges $d$, solve LP~\eqref{eq:lp} exactly. Its coefficients are rational with polynomial encoding length, so $r_d$ is a rational number of polynomial bit length that standard rational linear-programming methods compute exactly in polynomial time~\cite{schrijver1986theory}. One maximum-weight matching computation then recovers an optimal $M\ni d$, and the best pair over all $d$ is returned. By Lemma~\ref{lem:lottery}, $q_i(v)=c(M^\star,d^\star)\,p_{e(i)}$ if $i$ is matched by an edge $e(i)\in M^\star$, and $q_i(v)=0$ otherwise. Payments require one further optimization per selected advertiser (Section~\ref{sec:truth}).
\end{proof}

\paragraph{Computational model.} Optimization (Theorem~\ref{thm:poly}) runs in deterministic polynomial time in the Turing model with respect to the input bit length, and payments require at most $K$ further optimizations. Sampling uses at most $K$ independent Bernoulli draws whose parameters $a_k$ are rationals of polynomial bit length. To draw a Bernoulli$(a)$ variable exactly for rational $a$, we generate fair bits $\beta_1,\beta_2,\ldots$ and compare them with the binary digits of $a$, which long division produces one at a time, until the two sequences first differ, and output $1$ if the random bit is the smaller one. Each comparison stops with probability $1/2$, so a draw uses two bits in expectation, and each digit costs time polynomial in the bit length of $a$ and the current index. Sampling therefore takes expected polynomial time; we claim no worst-case bound. Ties in Step~1 may be broken by any fixed rule, since the monotonicity proof holds for every exact maximizer; this is also why exact rational arithmetic matters. Appendix~\ref{app:opt} discusses practical implementation.

\section{A Factor-Two Approximation}\label{sec:approx}

For any fixed matching, an adjacent-swap argument shows that displaying the ads in nonincreasing order of value maximizes cascade welfare. Algorithm~1 does not use this order: by Lemma~\ref{lem:lottery}, its expected allocation does not depend on the order of the non-designated ads.

We compare $\OPT$ with the symmetric benchmark
\[
 r_0(v)=\max_{M\in\calM_K}F_v(M),\qquad F_v(M)=\frac{A_v(M)}{1+S(M)},
\]
and then apply~\eqref{eq:symmetric}. Fix a feasible display $A$ with edges $e_1,\dots,e_\ell$ in this order, and write $p_k=p_{e_k}$ and $v_k=v_{e_k}$. Let $\rho_k=\prod_{h<k}(1-p_h)$ be the probability that the user reaches $e_k$, so that $\W(A)=\sum_k\rho_kv_kp_k$. The key observation is that $\W(A)$ equals the expected standalone welfare of a random prefix of the display. Let $P_t=\{e_1,\dots,e_t\}$ be drawn with probability
\[
 \lambda_t=\rho_t-\rho_{t+1}=\rho_tp_t\ \ (t<\ell),\qquad \lambda_\ell=\rho_\ell.
\]
\begin{figure}[t]
\centering
\includegraphics[width=\columnwidth]{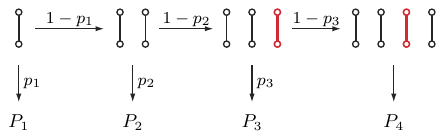}
\Description{Four matching-prefix states contain one, two, three, and four separate edges. At each of the first three states, a downward arrow returns the current prefix Pt with conditional probability pt, while a horizontal arrow appends the next edge with probability one minus pt. The fourth state returns P4 unconditionally. The third edge is red in both P3 and P4, so its inclusion probability is lambda3 plus lambda4, equal to rho3.}
\caption{Random-prefix decomposition for four matching edges, with $P_t$ the prefix of length $t$. Upon reaching $P_t$ ($t<4$), the process stops with probability $p_t$ or appends the next edge; $P_4$ is terminal. The red edge $e_3$ lies in $P_3$ and $P_4$, so it is included with probability $\lambda_3+\lambda_4=\rho_3$. This process is used only in the approximation proof, not by Algorithm~1.}
\label{fig:random-prefix}
\end{figure}
These weights form a distribution and satisfy $\sum_{t\ge k}\lambda_t=\rho_k$: edge $e_k$ belongs to the random prefix with exactly the probability that the user reaches it (Figure~\ref{fig:random-prefix}).

\begin{lemma}[Prefix decomposition]\label{lem:prefix}
The random prefix satisfies
\[
\textstyle\sum_t\lambda_tA_v(P_t)=\W(A)\quad\text{and}\quad\sum_t\lambda_tS(P_t)=1-\prod_{k\le\ell}(1-p_k)\le1.
\]
\end{lemma}
\begin{proof}
Exchanging the order of summation gives
\[
\textstyle\sum_t\lambda_tA_v(P_t)=\sum_kv_kp_k\sum_{t\ge k}\lambda_t=\sum_k\rho_kv_kp_k=\W(A).
\]
Similarly, $\sum_t\lambda_tS(P_t)=\sum_k\rho_kp_k=\sum_k(\rho_k-\rho_{k+1})=1-\rho_{\ell+1}$, where $\rho_{\ell+1}=\prod_{k\le\ell}(1-p_k)$.
\end{proof}

\begin{theorem}[Approximation]\label{thm:approx}
For every $v$, Algorithm~1 has expected cascade welfare $\Phi(v)$, and $\OPT(v)\le2\,r_0(v)\le2\,\Phi(v)$.
\end{theorem}
\begin{proof}
Each prefix $P_t$ belongs to $\calM_K$, so $A_v(P_t)\le r_0(v)\,(1+S(P_t))$ by the definition of $r_0$. Averaging over $t$ with weights $\lambda_t$ and applying Lemma~\ref{lem:prefix} gives $\W(A)\le r_0(v)(1+1)=2r_0(v)$. Since the display was arbitrary, $\OPT(v)\le2r_0(v)$. Finally, $r_0(v)\le\Phi(v)$ by~\eqref{eq:symmetric}, and $\Phi(v)$ is the expected cascade welfare of Algorithm~1 by Lemma~\ref{lem:lottery}.
\end{proof}

The symmetric ratio serves only as a certificate; designating the last ad raises welfare on the tested synthetic instances (Section~\ref{sec:synthetic}) but not in the worst case.

\begin{proposition}[Factor 2 is tight for the designated-last range]\label{prop:tight}
For every $\delta>0$ there is an instance with $\OPT(v)>(2-\delta)\,\Phi(v)$.
\end{proposition}
\begin{proof}
Fix a CTR level $\eta\in(0,1)$ and take $n$ advertisers and $n$ positions with $K=n$, $p_{ii}=\eta$, $p_{ij}=0$ for $i\ne j$, $v_1=1/\eta$, and $v_i=1$ for $i\ge2$. A matching that contains $(1,1)$ and $\ell-1$ other diagonal edges has $A_v=1+(\ell-1)\eta$; designating a diagonal edge gives denominator $1+(\ell-1)\eta$ and ratio $1$, and designating a zero-CTR edge only enlarges the denominator. Matchings without $(1,1)$ have ratio at most $\ell\eta/(1+(\ell-1)\eta)\le1$. Hence $\Phi=1$. Showing advertiser~1 first and then the other diagonal edges gives
\[
 \OPT\ \ge\ 1+\sum_{k=2}^n\eta(1-\eta)^{k-1}=1+(1-\eta)\bigl(1-(1-\eta)^{n-1}\bigr),
\]
which tends to $2$ as $n\to\infty$ and then $\eta\to0$.
\end{proof}
In this instance the optimal display shows the high-value, low-CTR ad first and then collects almost one more unit of welfare from many low-value ads, whereas every lottery in the range scales all selected ads by the same factor. Proposition~\ref{prop:tight} concerns our range only; it is not a lower bound for monotone mechanisms in general.

\section{Monotonicity and Payments}\label{sec:truth}

For each pair $(M,d)$, the lottery $D_{M,d}$ depends only on the public CTRs and the public order. The lottery range
\[
 \mathcal D=\{D_{M,d}:M\in\calM_K\setminus\{\emptyset\},\,d\in M\}\cup\{D_\emptyset\}
\]
is therefore independent of the bids; $D_\emptyset$ is the lottery that displays no ad. Write $q_i(X)$ for advertiser $i$'s expected cascade CTR under $X\in\mathcal D$. By Lemma~\ref{lem:lottery}, Algorithm~1 selects
\begin{equation}\label{eq:maxrange}
 X(v)\in\arg\max_{X\in\mathcal D}\ \sum_iv_iq_i(X),
\end{equation}
that is, it maximizes the true expected cascade welfare over $\mathcal D$. The monotonicity argument below is the standard one for exact maximization over a fixed range; it applies here because Lemma~\ref{lem:lottery} makes the optimized objective equal to the delivered welfare.

\begin{theorem}[Monotonicity]\label{thm:mono}
Fix $v_{-i}$ and two bids $v_i<v_i'$ for advertiser $i$, and let $X$ and $X'$ be any exact maximizers in~\eqref{eq:maxrange} at $(v_i,v_{-i})$ and $(v_i',v_{-i})$, respectively. Then $q_i(X')\ge q_i(X)$.
\end{theorem}
\begin{proof}
Let $B(X)=\sum_{h\ne i}v_hq_h(X)$. By optimality, $v_iq_i(X)+B(X)\ge v_iq_i(X')+B(X')$ and $v_i'q_i(X')+B(X')\ge v_i'q_i(X)+B(X)$. Adding the two inequalities gives $(v_i'-v_i)\bigl(q_i(X')-q_i(X)\bigr)\ge0$.
\end{proof}

The comparison is global, covering changes of the matching, the designated edge, and ties without discretizing bids; it requires only that~\eqref{eq:maxrange} maximizes the delivered welfare.

\paragraph{Payments.} Let $H_i(v_{-i})=\max_{X\in\mathcal D}\sum_{h\ne i}v_hq_h(X)=\Phi(0,v_{-i})$, and charge the expected payment
\begin{equation}\label{eq:payment}
 T_i(v)=H_i(v_{-i})-\sum_{h\ne i}v_hq_h(X(v))=H_i(v_{-i})-\Phi(v)+v_iq_i(v).
\end{equation}
This is the pivot payment for the range $\mathcal D$, and it coincides with the normalized single-parameter envelope payment. Indeed, for fixed $v_{-i}$, the function $V_i(w)=\Phi(w,v_{-i})$ is a finite maximum of affine functions of $w$, so it is convex and piecewise linear, and the selected $q_i(w,v_{-i})$ is a subgradient at every $w$. Hence $V_i(v_i)-V_i(0)=\int_0^{v_i}q_i(w,v_{-i})\,dw$, and substituting $V_i(v_i)=\Phi(v)$ and $V_i(0)=H_i(v_{-i})$ into~\eqref{eq:payment} gives $T_i(v)=v_iq_i(v)-\int_0^{v_i}q_i(w,v_{-i})\,dw$. Unselected advertisers pay nothing: if $q_i(v)=0$, then $\Phi(v)=\sum_{h\ne i}v_hq_h(X(v))\le H_i(v_{-i})\le\Phi(v)$, so $T_i(v)=0$. Payments therefore require one additional optimization per selected advertiser.

\begin{proposition}[Incentives]\label{prop:tie}
With payments~\eqref{eq:payment}, truthful reporting maximizes each advertiser's expected utility for every $v_{-i}$ (truthfulness in expectation), and $0\le T_i(v)\le v_iq_i(v)$ (individual rationality in expectation).
\end{proposition}
\begin{proof}
Suppose that advertiser $i$ has value $\theta_i$ and reports $b_i$, and let $X=X(b_i,v_{-i})$. Its expected utility is $\theta_iq_i(X)+\sum_{h\ne i}v_hq_h(X)-H_i(v_{-i})$. The last term does not depend on $b_i$, and by~\eqref{eq:maxrange} the sum of the first two is maximized by reporting $b_i=\theta_i$. For individual rationality, $T_i\ge0$ because $H_i$ maximizes the other advertisers' welfare over $\mathcal D$, and $T_i\le v_iq_i$ because $\Phi(v)\ge H_i(v_{-i})$ whenever $v_i\ge0$.
\end{proof}

\paragraph{Per-click billing.} The platform may instead charge $T_i(v)/q_i(v)$ per realized click, and nothing if $q_i(v)=0$. The expected payment, and hence the incentive guarantee, is unchanged, and the price per click never exceeds $v_i$, so a truthful advertiser never pays more than its value for a click. The mechanism is not universally truthful: incentives hold in expectation over its internal randomness.

We can now state the main result formally.
\begin{theorem}[Formal version of Theorem~\ref{thm:main}]\label{thm:main-formal}
Consider the pair-specific first-click cascade model with rational public standalone CTRs $p_{ij}\in[0,1]$, rational per-click values $v_i\ge0$, and any cardinality bound $K$. Algorithm~1, with the payments~\eqref{eq:payment}, has the following properties.
\begin{enumerate}[leftmargin=1.3em,itemsep=1pt,topsep=2pt]
\item \emph{Approximation:} $\E[\W(A(v))]\ge\frac12\OPT(v)$, where $A(v)$ is the random display produced for bids $v$ and $\OPT(v)$ is the optimal cascade welfare over all feasible matchings and display orders.
\item \emph{Monotonicity:} for every advertiser $i$, the expected cascade CTR $q_i(v)=\E[\pi_i(A(v))]$ is nondecreasing in $v_i$.
\item \emph{Incentives:} the mechanism is truthful in expectation and individually rational in expectation for risk-neutral, quasi-linear advertisers whose only private information is $v_i$.
\item \emph{Computation:} the selected lottery, the expected allocation $q(v)$, and the expected payments are computable in deterministic polynomial time. Sampling a display requires at most $K$ independent Bernoulli draws with rational parameters of polynomial bit length, and hence expected polynomial time with fair random bits.
\end{enumerate}
\end{theorem}

\begin{proof}
Approximation is Theorem~\ref{thm:approx}, monotonicity is Theorem~\ref{thm:mono}, incentives are Proposition~\ref{prop:tie}, and computation follows from Theorem~\ref{thm:poly} and the computational model of Section~\ref{sec:opt}.
\end{proof}
\section{Why Exact Surrogate Optimization Is Not Enough}\label{sec:counterexample}

The constant-factor algorithm of Balseiro et al.~\cite{balseiro2026position} approximates a \emph{restricted welfare}, in which advertisers are taken in nonincreasing order of value and each matched ad contributes the minimum of $p_e$ and the remaining mass; the remaining mass starts at $1$ and decreases by each earlier contribution. This surrogate approximates cascade welfare within a factor of $4$. One might expect non-monotonicity to arise only from approximating the surrogate, so that exact optimization followed by display in value order would be monotone. The following instance shows otherwise.

Take $n=m=K=3$, bids $(v_1,v_2,v_3)=(2,1,u)$, and
\begin{equation}\label{eq:counterp}
 p=\begin{pmatrix}
 1/10&3/20&0\\
 0&2/5&1/4\\
 3/5&0&9/10
 \end{pmatrix}.
\end{equation}
For both bids $u\in\{\tfrac25,\tfrac35\}$ the value order is $1,2,3$. Let $M_A$ be the diagonal matching and $M_B=\{(1,2),(2,3),(3,1)\}$. Figure~\ref{fig:surrogate-switch} shows both matchings. Their surrogate values are $\tfrac35+\tfrac12u$ and $\tfrac{11}{20}+\tfrac35u$. Enumerating all $34$ feasible matchings (Table~\ref{tab:counter}) shows that $M_A$ is the unique surrogate optimum at $u=\tfrac25$ and $M_B$ is the unique optimum at $u=\tfrac35$. Writing $\pi_3(M)$ for advertiser~3's cascade CTR when $M$ is displayed in value order, this CTR nevertheless falls:
\begin{align*}
 \pi_3(M_A)&=\tfrac9{10}\bigl(1-\tfrac1{10}\bigr)\bigl(1-\tfrac25\bigr)=\tfrac{243}{500},\\
 \pi_3(M_B)&=\tfrac35\bigl(1-\tfrac3{20}\bigr)\bigl(1-\tfrac14\bigr)=\tfrac{153}{400}.
\end{align*}

\begin{figure}[tb]
\centering
\includegraphics[width=\columnwidth]{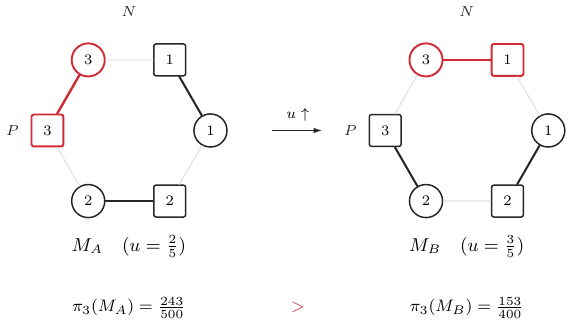}
\Description{Two panels show the positive-CTR advertiser-position pairs as a six-cycle. Circles are advertisers and rounded squares are positions. The left matching MA selects pairs (1,1), (2,2), and (3,3) at bid u equals two fifths. The right matching MB selects pairs (1,2), (2,3), and (3,1) at bid u equals three fifths. Selected edges are bold, and advertiser 3's edge is red. A rightward arrow marks the bid increase. The displayed inequality shows that advertiser 3's cascade CTR falls from 243/500 to 153/400.}
\caption{Exact restricted-welfare optimization is not monotone. Circles are advertisers, squares are positions, and gray edges are unselected positive-CTR pairs. Raising $u$ from $\tfrac25$ to $\tfrac35$ switches the unique surrogate optimum from $M_A$ to $M_B$, and advertiser~3's cascade CTR (red edge) falls.}
\label{fig:surrogate-switch}
\end{figure}

\begin{table}[t]
\centering
\caption{Monotonicity failure of the exact restricted-welfare optimizer (all 34 matchings enumerated).}
\label{tab:counter}
\small
\begin{tabular}{@{}ccccc@{}}
\toprule
$u$ & Unique opt. & Surrogate & Runner-up & True $\pi_3$\\
\midrule
$2/5$ & $M_A$ & $4/5$ & $79/100$ & $243/500$\\
$3/5$ & $M_B$ & $91/100$ & $9/10$ & $153/400$\\
\bottomrule
\end{tabular}
\end{table}

\begin{table*}[t]
\centering
\caption{Small-instance welfare$/\OPT$ (2,000 seeded instances per row), reported as mean / 5th percentile / minimum.}
\label{tab:synthetic}
\small
\begin{tabular}{@{}llcccc@{}}
\toprule
$(n,m,K)$ & Regime & \textbf{DL (ours, TIE)} & Sym (benchmark) & RW (non-monotone oracle) & Bkt (prior monotone) \\
\midrule
$(4,4,4)$ & low CTR & \textbf{.848 / .803 / .760} & .748 / .709 / .674 & 1.000 / 1.000 / .964 & .351 / .294 / .243 \\
 & uniform & \textbf{.926 / .840 / .728} & .606 / .518 / .500 & .956 / .866 / .763 & .447 / .359 / .230 \\
 & heterog. & \textbf{.866 / .775 / .714} & .583 / .523 / .501 & .967 / .897 / .816 & .457 / .353 / .271 \\
\addlinespace
$(6,4,3)$ & low CTR & \textbf{.886 / .857 / .828} & .776 / .750 / .728 & 1.000 / 1.000 / .965 & .389 / .341 / .262 \\
 & uniform & \textbf{.934 / .858 / .763} & .635 / .558 / .500 & .952 / .851 / .764 & .482 / .405 / .294 \\
 & heterog. & \textbf{.867 / .785 / .743} & .588 / .525 / .501 & .964 / .889 / .783 & .492 / .395 / .297 \\
\bottomrule
\end{tabular}
\end{table*}

Raising advertiser~3's bid from $0.4$ to $0.6$ thus lowers its cascade CTR from $0.486$ to $0.3825$. Both optima are unique, and neither the value order nor the number of displayed ads changes. The failure is therefore not due to tie-breaking, bucketing, or approximation: the click allocation of the surrogate differs from that of the cascade. The example concerns this surrogate only. It does not show that other deterministic or surrogate-based approaches must fail, nor that the PTAS-based allocation is non-monotone. It motivates our design, which makes the optimized objective equal to the delivered welfare (Lemma~\ref{lem:lottery}) before optimizing it.

\section{A Descriptive Synthetic Comparison}\label{sec:synthetic}

The factor $2$ is a worst-case guarantee. To see how it compares with welfare on specified synthetic distributions, we compare three allocation rules and one analytical benchmark with the exact optimum on small random instances (Section~\ref{sec:small}), run the polynomial-time optimizer on larger instances (Section~\ref{sec:scale}), and measure sensitivity to CTR estimation error (Section~\ref{sec:noise}). These experiments illustrate the theorems; they are not a statistical validation.

\subsection{Small instances with exact optimum}\label{sec:small}

We compare the following methods and benchmarks; RW and Bkt are taken from Balseiro et al.~\cite{balseiro2026position} and are tied to specific results there.
\begin{itemize}[leftmargin=1.3em,itemsep=1pt,topsep=2pt]
\item \textbf{DL}: Algorithm~1, with expected welfare $\Phi$.
\item \textbf{Sym}: the symmetric benchmark $r_0$ of Section~\ref{sec:approx}, an analytical quantity rather than a mechanism.
\item \textbf{RW} (non-monotone oracle): an \emph{exact} maximizer of the restricted welfare of~\cite[Sec.~5.2]{balseiro2026position}, the objective approximated by their PTAS~\cite[Thm.~5.12, Cor.~5.13]{balseiro2026position}, displayed in nonincreasing value order~\cite[Lemma~5.2]{balseiro2026position}. Exact enumeration makes it stronger than their polynomial-time allocation. It has no monotonicity guarantee and serves only as a welfare reference.
\item \textbf{Bkt} (prior monotone algorithm): our implementation of~\cite[Sec.~5.4, Thm.~5.14]{balseiro2026position}, which picks one of $L=\lceil\log_2(4m)\rceil$ CTR buckets, $b$, uniformly at random~\cite[Lemma~5.15]{balseiro2026position} and runs their greedy Algorithm~2 by $v_ip_{ij}$ with cardinality cap $\min(2^b,m)$, displaying ads in insertion order. Bucket $b<L$ contains the edges with $2^{-b}<p_{ij}\le2^{-(b-1)}$, and bucket $L$ the remaining edges with $p_{ij}>0$. We add the cap $K$, which among the small configurations binds only for $(6,4,3)$, and compute the exact expectation over buckets.
\end{itemize}

We use two sizes, $(n,m,K)=(4,4,4)$ and $(6,4,3)$, the latter with a binding cardinality limit, and three regimes: low CTR, with $p_{ij}\sim U[0,0.2]$ and $v_i\sim U[0.2,1]$; uniform, with $p_{ij}\sim U[0,1]$ and $v_i\sim U[0.2,1]$; and heterogeneous, with $p_{ij}\sim\mathrm{Beta}(0.7,2)$ and $\log v_i\sim N(0,0.8^2)$. Each of the six cells has 2,000 seeded instances, 12,000 in total. $\OPT$ is computed by enumerating all matchings, each displayed in value order. Table~\ref{tab:synthetic} reports the mean, 5th percentile, and minimum of welfare$/\OPT$.

\begin{figure}[t]
\centering
\includegraphics[width=\columnwidth]{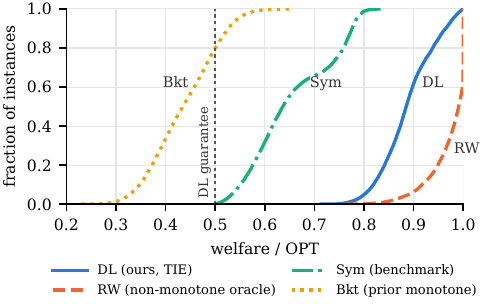}
\caption{Empirical CDFs of welfare$/\OPT$ over 12,000 synthetic instances (Section~\ref{sec:small}). The vertical line marks DL's worst-case guarantee of $1/2$.}
\Description{Empirical cumulative distribution functions of the welfare ratio for four methods. The prior monotone bucketing rule lies mostly between 0.3 and 0.6. The symmetric benchmark lies between 0.5 and 0.83. The designated-last mechanism lies between 0.71 and 1.0, mostly between 0.8 and 0.97. The non-monotone restricted-welfare oracle lies between 0.76 and 1.0, with a large mass at 1.0. A vertical line marks 0.5, the worst-case guarantee for the designated-last mechanism.}
\label{fig:ecdf}
\end{figure}

Across the six synthetic settings, DL achieves 85--93\% of $\OPT$ on average and improves on Sym by 10--32 percentage points. It also weakly outperforms Bkt on all 12,000 sampled instances, without implying dominance beyond the sample. RW achieves higher mean welfare but serves as a non-monotone oracle.

\begin{table*}[t]
\centering
\caption{Larger-instance results (20 seeded instances per row). Times are wall-clock seconds on one 2.1\,GHz cloud CPU (SciPy/HiGHS). LPs: median designated-edge LP calls plus one for $r_0$. Selected: mean $|M^\star|$ before thinning. Brackets: means of instance-wise numerical welfare$/\OPT$ endpoints.}
\label{tab:scale}
\small
\setlength{\tabcolsep}{4.5pt}
\begin{tabular}{@{}llcccccc@{}}
\toprule
$(n,m,K)$ & Regime & Alloc.\ (median) & Alloc.\ $+$ payments (median / max) & LPs & Selected & \textbf{DL/OPT bracket} & Bkt/OPT bracket \\
\midrule
$(20,5,3)$ & low CTR & 0.01 & 0.03 / 0.05 & 1 & 3.0 & \textbf{[0.74, 0.88]} & [0.32, 0.39] \\
 & heterog. & 0.01 & 0.03 / 0.05 & 1 & 2.5 & \textbf{[0.72, 0.86]} & [0.38, 0.46] \\
$(50,10,5)$ & low CTR & 0.03 & 0.17 / 0.23 & 2 & 5.0 & \textbf{[0.57, 0.83]} & [0.26, 0.38] \\
 & heterog. & 0.03 & 0.09 / 0.17 & 1 & 2.1 & \textbf{[0.70, 0.85]} & [0.38, 0.46] \\
$(100,10,5)$ & low CTR & 0.06 & 0.34 / 0.45 & 1 & 5.0 & \textbf{[0.57, 0.83]} & [0.26, 0.38] \\
 & heterog. & 0.08 & 0.26 / 0.66 & 1 & 2.2 & \textbf{[0.73, 0.87]} & [0.38, 0.46] \\
$(200,20,8)$ & low CTR & 0.41 & 3.60 / 5.32 & 2 & 8.0 & \textbf{[0.54, 0.80]} & [0.25, 0.37] \\
 & heterog. & 0.56 & 1.64 / 3.30 & 1 & 2.0 & \textbf{[0.68, 0.88]} & [0.37, 0.47] \\
\bottomrule
\end{tabular}
\end{table*}

\subsection{Larger instances: running time and welfare brackets}\label{sec:scale}

A response typically has many candidate advertisers and few insertion points, so we next ran the polynomial-time optimizer of Section~\ref{sec:opt}, rather than enumeration, on the four sizes of Table~\ref{tab:scale} with 20 seeded instances per size in the low-CTR and heterogeneous regimes. The implementation solves LP~\eqref{eq:lp} in floating point with HiGHS through SciPy and uses two shortcuts that preserve the exact mathematical optimum (Appendix~\ref{app:opt}). Before timing, we checked the optimizer against exhaustive enumeration on 300 small instances and against the unpruned search on 20 medium instances.

$\OPT$ cannot be enumerated at these sizes, so Table~\ref{tab:scale} reports welfare brackets. On one side, $\OPT\le\mathrm{UB}=\min\{2r_0,\mathrm{MW}\}$: the bound $2r_0$ is Theorem~\ref{thm:approx}, and $\mathrm{MW}=\max_{M\in\calM_K}A_v(M)$, the largest standalone welfare, attained by a matching $M_{\mathrm{MW}}$, bounds $\OPT$ because cascade CTRs never exceed standalone CTRs. On the other side, $\OPT\ge\mathrm{LB}$, the largest value-ordered cascade welfare among the matchings selected by DL, $M_{\mathrm{MW}}$, the symmetric benchmark, and the Bkt buckets. Hence $\Phi/\mathrm{UB}\le\mathrm{DL}/\OPT\le\Phi/\mathrm{LB}$, and similarly for Bkt.

At the largest tested size, median allocation time remains below one second, and the maximum time including payments is 5.32\,s. DL achieves higher welfare than Bkt on all 160 sampled instances. Under heterogeneous CTRs, DL typically selects fewer ads because additional CTR mass reduces $c(M,d)$.

\subsection{Sensitivity to CTR estimation error}\label{sec:noise}

In practice, $p_{ij}$ may be estimated by a prediction model. To measure welfare sensitivity to misspecification, we computed each rule from noisy estimates $\hat p_{ij}=\min\{1,\,p_{ij}e^{\sigma Z_{ij}}\}$ with independent standard normal $Z_{ij}$, and evaluated its welfare under the true $p$ against the true $\OPT$. We used $(n,m,K)=(6,4,3)$, where $\OPT$ can be enumerated, with 1,000 seeded instances per regime; all noise levels share these instances, and perturbations are drawn independently for each nonzero $\sigma$. For DL, the retention probabilities are computed from $\hat p$, while true clicks follow $p$; its expected cascade CTRs still follow exactly from the recursion $\E[\pi_k]=R_ka_kp_k$, $R_{k+1}=R_k(1-a_kp_k)$, so no click simulation is needed. We evaluate welfare only; Section~\ref{sec:discussion} discusses incentives.

\begin{table}[t]
\centering
\caption{True welfare$/\OPT$ under noisy CTR estimates (mean / 5th percentile; 1,000 instances per regime; $(n,m,K)=(6,4,3)$). $\sigma$ is the log-space noise standard deviation.}
\label{tab:noise}
\small
\setlength{\tabcolsep}{4pt}
\begin{tabular}{@{}llccc@{}}
\toprule
Regime & $\sigma$ & \textbf{DL} & RW (oracle) & Bkt \\
\midrule
low CTR & $0.0$ & \textbf{0.89 / 0.86} & 1.00 / 1.00 & 0.39 / 0.34 \\
 & $0.1$ & \textbf{0.88 / 0.83} & 0.99 / 0.94 & 0.39 / 0.34 \\
 & $0.3$ & \textbf{0.83 / 0.72} & 0.94 / 0.81 & 0.44 / 0.34 \\
 & $0.5$ & \textbf{0.78 / 0.61} & 0.88 / 0.68 & 0.47 / 0.35 \\
\addlinespace
heterog. & $0.0$ & \textbf{0.87 / 0.78} & 0.96 / 0.89 & 0.49 / 0.39 \\
 & $0.1$ & \textbf{0.86 / 0.76} & 0.96 / 0.88 & 0.49 / 0.39 \\
 & $0.3$ & \textbf{0.82 / 0.66} & 0.93 / 0.75 & 0.49 / 0.38 \\
 & $0.5$ & \textbf{0.78 / 0.53} & 0.88 / 0.62 & 0.50 / 0.37 \\
\bottomrule
\end{tabular}
\end{table}

Table~\ref{tab:noise} shows that DL degrades gradually, losing about ten points of mean welfare at $\sigma=0.5$, comparable to the RW oracle, and it remains well above Bkt. The guarantee, however, holds relative to the estimates: at $\sigma=0.5$ the smallest observed ratio is $0.26$, below $1/2$, so the factor $2$ should be read as conditional on accurate CTRs. Bkt's mean rises with noise mainly in the low-CTR regime, where noise leaves fewer of its uniformly chosen buckets empty; this reflects bucket occupancy, not robustness to misestimation.

\section{Limitations and Future Work}
\label{sec:discussion}

Our mechanism assumes a generative advertising pipeline in which a generator proposes insertion points and a prediction model provides public CTRs $p_{ij}$ for eligible advertiser--position pairs. The mechanism computes $(M^\star,d^\star)$ and payments, thins the non-designated ads, and renders the retained ads in the prescribed order with $d^\star$ last. Per-click billing is supported (Section~\ref{sec:truth}), and per-section caps $K_g$ preserve our guarantees (Section~\ref{sec:opt}).

Three limitations delimit these guarantees. First, the public CTRs must correctly describe user click behavior. When $\hat p\ne p$, the welfare guarantee relative to the true optimum and monotonicity of actual click allocations may fail, so truthfulness in expectation for actual-click utility is not guaranteed (Table~\ref{tab:noise}). Second, the platform must control display order without changing standalone CTRs; fixed-order layouts are not covered by our designated-last construction. Third, random thinning may leave some selected ads undisplayed, and incentive guarantees hold only in expectation, not for every realization. On the tested synthetic instances, our mechanism achieves lower average welfare than the non-monotone oracle (Table~\ref{tab:synthetic}), although this gap does not establish an inherent cost of truthfulness.

Future directions include robust mechanism design under CTR misspecification, strategic creatives, and feedback between allocation and learning~\cite{turnersmith2026pricing}. It also remains open whether polynomial-time monotone mechanisms can improve on factor $2$, which is tight only for our lottery range, and whether ranges that randomize the last ad admit exact polynomial-time optimization (Section~\ref{sec:mechanism}). Further questions concern revenue-optimal design and propensity-based evaluation using logged retention probabilities; evaluating alternative allocation rules beyond the lottery's support requires additional overlap assumptions.

\section{Reproducibility}\label{sec:reproducibility}
The anonymous artifact\footnote{\url{https://anonymous.4open.science/r/anonymous-review-artifact-5061/}} provides code and data for all experiments, together with a Lean~4 formalization that machine-checks key steps in the theoretical proofs.

\clearpage
\bibliographystyle{ACM-Reference-Format}
\bibliography{references}

\appendix

\section{Implementation Details}

\paragraph{Practical implementation.}\label{app:opt} Dinkelbach's parametric method~\cite{dinkelbach1967fractional}, which alternates the recovery step of Section~\ref{sec:opt} with updates of $r_d$, can implement Step~1 of Algorithm~1; it does not replace the bound of Theorem~\ref{thm:poly}. The experiments of Section~\ref{sec:scale} use two shortcuts that preserve exactness. First, by the observation in Section~\ref{sec:mechanism}, the LP for $d$ can be restricted to edges with $p_e\le p_d$. Second, $d$ is skipped when a bound certifies $r_d<r^\ast$, the best ratio found so far, initialized at $r^\ast=r_0\le\Phi$. By the equivalence used in the recovery step, $r_d>r^\ast$ requires $v_dp_d-r^\ast+\sum_{e\in M\setminus\{d\}}(v_e-r^\ast)p_e>0$ for some $M\ni d$, and the sum is at most the total of the $K-1$ largest per-advertiser maxima of $\max\{v_e-r^\ast,0\}\,p_e$.

\section{Counterexample Certificate}\label{app:counter}

For~\eqref{eq:counterp}, the surrogate values of $M_A$ and $M_B$ are affine in $u$ and cross at $u=1/2$. At $u=2/5$, exact enumeration of all 34 feasible matchings gives the optimum $4/5$ at $M_A$ and the runner-up value $79/100$; at $u=3/5$, it gives the optimum $91/100$ at $M_B$ and the runner-up value $9/10$. The artifact cited in Section~\ref{sec:reproducibility} reproduces this certificate in exact rational arithmetic.

\end{document}